\documentclass[journal]{IEEEtran}

\IEEEoverridecommandlockouts

\usepackage{wrapfig}
\usepackage{amsthm}
\usepackage{mathrsfs}
\usepackage{graphicx}
\usepackage{graphics}
\usepackage{amssymb}
\usepackage{amsmath,mathtools}
\usepackage{color}
\usepackage{xspace}
\usepackage{algpseudocode}
\usepackage{bbm}
\usepackage{comment}
\usepackage{algorithmicx}
\usepackage[caption=false]{subfig}
\usepackage{psfrag}
\usepackage{url}
\usepackage{float}
\usepackage{soul}

\usepackage{rotating}
\usepackage{chngcntr}
\usepackage{apptools}
\AtAppendix{\counterwithin{thm}{section}}
\usepackage{graphicx}
\usepackage{graphics}
\usepackage{amssymb}
\usepackage{amsmath}
\usepackage{mathtools}
\usepackage{color}
\usepackage{xspace}
\usepackage{algpseudocode}
\usepackage{bbm}
\usepackage{comment}
\usepackage{algorithmicx}
\usepackage{subfig}
\usepackage{psfrag}
\usepackage{tabu}

\usepackage{amsmath}
\allowdisplaybreaks[4]
\usepackage{tikz}
\usetikzlibrary{shapes,arrows}
\usetikzlibrary{positioning}

\tikzstyle{block}=[draw opacity=0.7,line width=1.4cm]
\DeclareMathAlphabet{\mathpzc}{OT1}{pzc}{m}{it}
\definecolor{CranJ}{cmyk}{0,0.69,0.54,0.04} 
\definecolor{PinkJ}{cmyk}{0,0.71,0.43,0.12} 
\definecolor{Cran}{cmyk}{0,0.73,0.41,0.29} 
\definecolor{VRed}{cmyk}{0,0.75,0.25,0.2} 
\definecolor{ORed}{cmyk}{0,0.75,0.75,0} 
\definecolor{CBlue}{cmyk}{1,0.25,0,0} 

\title{\LARGE \bf 
Distributed Leader Following of an Active Leader for Linear Heterogeneous Multi-Agent Systems
}

\author{Yi-Fan Chung and Solmaz S. Kia %
  \thanks{The authors are with the Department of Mechanical and Aerospace Engineering, University of California Irvine, Irvine, CA 92697,  
    {\tt\small \{yfchung,solmaz\}@uci.edu}. This work is supported by NSF award IIS-SAS-1724331.}%
}

\newcommand{\ee}{\operatorname{e}}

\newcommand{\VV}{\mathcal{V}}
\newcommand{\EE}{\mathcal{E}}
\newcommand{\GG}{\mathcal{G}}

\newcommand{\real}{{\mathbb{R}}}
\newcommand{\reals}{{\mathbb{R}}}

\newcommand{\realpositive}{{\mathbb{R}}_{>0}}
\newcommand{\realnonnegative}{{\mathbb{R}}_{\ge 0}}

\newcommand{\argmin}{\operatorname{argmin}}

\newcommand{\dout}{\mathsf{d}_{\operatorname{out}}}

\newcommand{\until}[1]{\in\{1,\dots,#1\}}

\newcommand{\vect}[1]{\boldsymbol{\mathbf{#1}}}
\newcommand{\vectsf}[1]{\vect{\mathsf{#1}}}
\newcommand{\Bvect}[1]{\bar{\boldsymbol{\mathbf{#1}}}}

\newcommand{\dvect}[1]{\dot{\vect{#1}}}

\newcommand{\Diag}[1]{\operatorname{Diag}(#1)}

 \newcommand{\boxend}{\hfill \ensuremath{\Box}}

\newtheorem{thm}{Theorem}[section]

\newtheorem{rem}{Remark}[section]
\newtheorem{cor}{Corollary}[section]
\newtheorem{lem}{Lemma}[section]

\newcommand{\oprocendsymbol}{\hbox{$\bullet$}}
\newcommand{\oprocend}{\relax\ifmmode\else\unskip\hfill\fi\oprocendsymbol}

\begin{document}


\maketitle
\begin{abstract}
This paper considers a leader-following problem for a group of heterogeneous linear time invariant (LTI) followers that are interacting over a directed acyclic graph. Only a subset of the followers has access to the state of the leader in specific sampling times. The dynamics of the leader that generates its sampled states is unknown to the followers. For interaction topologies in which the leader is a global sink in the graph, we propose a distributed algorithm that allows the followers to arrive at the sampled state of the leader by the time the next sample arrives. Our algorithm is a practical solution for a leader-following problem when there is no information available about the state of the leader except its instantaneous value at the sampling times.  Our algorithm also allows the followers to track the sampled state of the leader with a locally chosen offset that can be time-varying. When the followers are mobile agents whose state or part of their state is their position vector, the offset mechanism can be used to enable the followers to form a transnational invariant formation about the sampled state of the leader.
We prove that the control input of the followers to take them from one sampled state to the next one is minimum energy. We also show in case of the homogeneous followers, after the first sampling epoch the states and inputs of all the followers are synchronized with each other. Numerical examples demonstrate our results.
\end{abstract}


\section{Introduction}
Synchronization of multi-agent systems (MASs) is an important component of  many cooperative control problems, such as rendezvous~\cite{JL-ASM-BDOA:03}, formation control~\cite{WR:07}, flocking control~\cite{HS-XW-ZL:09}, containment control~\cite{YC-WR-ME:12} and sensor networks~\cite{ROS-JSS:05}. Synchronization problems can be roughly categorized into leaderless and leader-following. In the leaderless synchronization, which is closely related to the consensus problem, the agents aim to reach to a static or dynamic agreement on a common value~\cite{ROS-RMM:04,DVD-KJK:07,SSK-BVS-JC-RAF-KML-SM-arxiv}.
On the other hand, in the leader-following synchronization, agents aim to make the agreement on the states generated by a leader. In this paper, we focus on the design of a distributed leader-following algorithm when the only information available about the leader is its sampled state, which is only available to a subset of followers.

\emph{Literature review}:
The leader-following algorithms for single integrator and double integrator dynamics are presented in~\cite{WR-RWB:08},  and for homogeneous LTI systems 
are proposed in~\cite{WN-DC:10} and \cite{HZ-FLL-AD:11}.
For systems constituted of heterogeneous LTI followers, \cite{XW-YH:10} and \cite{YS-JH:12} propose the algorithms to synchronize with a passive zero-input LTI leader. 
\cite{YH-JH-LG:06} and \cite{YH-GC-LB:08} develop the controls for the single and double integral system, respectively, to track an active leader (active leader is a leader that has a control input).  But their works assume the leader's control input is available to all the followers.
\cite{LZ-RW:13} and \cite{YT-YH-XW:15} propose a leader-following algorithm respectively for homogeneous LTI and heterogeneous nonlinear MASs in which the unknown input of the leader is bounded and is not available to any follower. But the control inputs in \cite{LZ-RW:13} and \cite{YT-YH-XW:15} have the sliding mode structure and suffer from the well-known undesirable chattering behavior. We recall that from a practical perspective, chattering is undesirable and leads to excessive control energy expenditure~\cite{SlotineLi:91}. \cite{YY-MH:18} is the recent result for the leader-following problem, which is based on the result of \cite{LZ-RW:13} and develops a distributed observer to estimate the leader's state for each follower. Then, the output synchronization of heterogeneous leader-follower linear systems is achieved by optimal local tracking of the output of the observer. We note that in both~\cite{LZ-RW:13} and~\cite{YY-MH:18}, the active leader is restricted to be linear and have limited input. The work reviewed so far are all converge to leader following in an asymptotic manner, i.e., the settling time to reach an agreement is infinity. For fast convergence, \cite{LW-FX:10},\cite{SL-HD-XL:11} and \cite{YC-WR:12} propose the finite-time synchronization algorithms for single and double integral MASs, where the upper bound of the settling time explicitly depends on the initial state of the MAS. Therefore, to use these algorithms, the centralized knowledge of the initial state of the MAS is essential to estimate the settling time.  \cite{MD-AP-GD-MD-KV:15} and \cite{ZZ:15} propose the fixed-time synchronization algorithms, where the settling time is bounded and independent of the initial state of the MAS. However, for both these finite and fixed-time algorithms, the settling time is upper bounded by a conservative estimation. \cite{YZ-YL-GW-WR-GC:18} introduces the specified-time synchronization control for the leaderless MASs in which one can determine the settling time exactly in advance. Specified-time synchronization can be useful to the applications that require precise acting time, such as target attack at a specified time.

\emph{Statement of contributions}: 
In this paper, we consider a leader-following problem in which the only information available about the leader is its instantaneous sampled state that is known only to a subset of a group of heterogeneous LTI followers at the sampling times. We make no assumptions about the input of the leader or the structural form of its dynamics. That is, the state of the leader is perceived by the followers as an exogenous signal.  The sampled states of the leader can be the states of a physical system (e.g., in a pursuit-evasion problem) or a set of desired reference states of a virtual leader (e.g., in a waypoint tracking problem). Given the limited information about the leader, we seek a practical solution that enables the followers to arrive at the sampled state of the leader before the next sampling time. That is, we design a distributed algorithm that steers a group of heterogeneous LTI followers to be at the sampled states of the leader at finite time just before the next sampled state is obtained. We note that practical one step lagged tracking has also been used in~\cite{MZ-SM:10,EM-JIM-CS-SM:14,EM-JIM-CS-SM:14-2} for a set of dynamic average consensus algorithms with asymptotic tracking behavior.  Our solution is inspired by the minimum energy controller design~\cite{FL-DV-VS:12} in the classical optimal control theory, and is proposed for problems where the interaction topology of the followers plus the leader is an acyclic digraph with the leader as the global sink. Directed acyclic interaction topology can be interpreted as the agents only obtaining information from those in front of them (see, \cite{HGT-GJP-VK:04,DW-YG-LZ:10} for algorithms designed over acyclic graphs). Our algorithm also allows the followers to track the sampled state of the leader with a locally chosen offset, which can be time-varying. This offset, when the followers are mobile agents and their whole state or part of it is the position vector, can be used to enable the followers to form a transnational invariant formation~\cite{MM-EM::10} about the sampled state of the leader.  For a special class of non-homogeneous LTI MAS, we show that our results can be extended to solve a leader-following problem where we want only an output of the followers to follow the leader's sampled state. Finally, we show that if the followers are homogeneous, our algorithm not only results in a leader following behavior, but also it makes the states and inputs of the followers fully synchronized after the first sampling epoch. We demonstrate our leader-following results via three numerical examples. In the first example, we show the application of our leader-following algorithm in following a nonlinear mass-spring-damper leader under a specific formation structure for a group of heterogeneous linear mass-spring-damper systems. In the second example, we demonstrate the use of our algorithm for reference state tracking via a group of second order integrator followers with bounded control. The result shows the synchronization of the homogeneous followers is realized. Moreover, using the intrinsic properties of our leader following algorithm, we show that the arrival times at the reference states can be specified in such a way that the inputs of the followers stay within the saturation bounds. Our last example demonstrates an output-tracking scenario for a group of aircraft.

\emph{Organization}:
The rest of this parer is organized as follows. Section~\ref{se::notations} gathers basic notation and graph-theoretic notions. Section~\ref{sec::prob_def} gives our problem definition and objective statement. Section~\ref{sec::main} proposes our distributed leader-following algorithm. In Section~\ref{sec::demo}, three applications are demonstrated.  Section~\ref{sec::conclusion} concludes the results of this paper. Appendix A contains the proof of our main result, Theorem~\ref{thm::main}. Finally, Appendix B presents an auxiliary result, which is invoke to support the feasibility of the sampling time design in our second numerical example.

\section{Notations}\label{se::notations}

\emph{Notation}: We let $\reals$, $\realpositive$, $\realnonnegative$, $\mathbb{Z}$, and $\mathbb{Z}_{\geq 0}$
denote the set of real, positive real, non-negative real, integer, and non-negative integer numbers, respectively.
The transpose of a matrix $\vect{A}\in\real^{n\times m}$ is~$\vect{A}^\top$.

\emph{Graph theoretic notations and definitions}: Here we review our graph related notations and relevant definitions and concepts from graph
theory following~\cite{FB-JC-SM:09}. A \emph{digraph}, is a triplet $\GG = (\VV ,\EE,\vect{\mathsf{A}})$,~where $\VV=\{1,\dots,N\}$ is the \emph{node set} and
$\EE \subseteq \VV\times \VV$ is the \emph{edge set}, and $\vect{\mathsf{A}}=[\mathsf{a}_{ij}]\in\real^{N\times N}$ is the \emph{adjacency}
matrix of the graph defined according to $\mathsf{a}_{ij} =1$ if $(i, j) \in\EE$ and $
\mathsf{a}_{ij} = 0$, otherwise.  An edge $(i, j)$ from
$i$ to $j$ means that agent $j$ can send
information to agent $i$. Here,  $i$ is called an
\emph{in-neighbor} of $j$ and $j$ is called an \emph{out-neighbor}
of~$i$.  
A \emph{directed path} is a sequence of nodes
connected by edges. A directed path that starts and ends at the same node and all other nodes on the path are distinct is called a \emph{cycle}. A digraph without cycles is called \emph{directed acyclic graph}. The
\emph{out-degree} of a node $i$ is
 $\dout^i=\Sigma^N_{j =1} \mathsf{a}_{ij}$. The out-degree matrix of a graph is $\vectsf{D}_{\text{out}}=\Diag{\dout^1,\dout^2,\cdots,\dout^N}$.
We denote the set of in-neighbors of an agent $i$ by $\mathcal{N}^i_{\text{in}}$ and the out-neighbors of agent $i$ by $\mathcal{N}^i_{\text{out}}$. A node $i \in \VV$ is called a \emph{global sink} of $\GG$ if it outdegree $\dout^i=0$ and for every node $j \in \VV$ there is at least a path from $j$ to $i$.

\begin{figure}
    \centering
     \begin{tikzpicture}[auto,thick,scale=0.5, every node/.style={scale=0.67}]
\tikzset{edge/.style = {->,> = latex'}}
el/.style = {inner sep=2pt, align=left, sloped},
\node (0) at (0,0) [draw, minimum size=15pt,color=blue, circle, very thick,fill=red!30] {{\small \textbf{0}}};

\node (1) at (-1.5,1.5) [draw, minimum size=15pt,color=blue, circle, very thick] {{\small \textbf{1}}};

\node (2) at (-1.5,-1.5) [draw, minimum size=15pt,color=blue, circle, very thick] {{\small \textbf{2}}};

\node (3) at (-3,0) [draw, minimum size=15pt,color=blue, circle, very thick] {{\small \textbf{3}}};

\node (4) at (-4.25,0.5) [draw, minimum size=15pt,color=blue, circle, very thick] {{\small \textbf{4}}};

\node (5) at (-4.5,-0.5) [draw, minimum size=15pt,color=blue, circle, very thick] {{\small \textbf{5}}};

\node (6) at (-6,1.5) [draw, minimum size=15pt,color=blue, circle, very thick] {{\small \textbf{6}}};

\node (7) at (-6,-1.5) [draw, minimum size=15pt,color=blue, circle, very thick] {{\small \textbf{7}}};

\draw[edge,dashed]  (1)to  (0);
\draw[edge,dashed]  (2)to  (0);
\draw[edge,dashed]  (3)to  (0);
\draw[edge]  (3)to  (1);
\draw[edge]  (3)to  (2);
\draw[edge]  (4)to  (1);
\draw[edge]  (5)to  (3);
\draw[edge]  (6)to  (1);
\draw[edge]  (6)to  (4);
\draw[edge]  (7)to  (5);
\draw[edge]  (7)to  (2);

\end{tikzpicture}
    \caption{A leader-follower network. The interaction topology of the follower agents, $\mathcal{G}$, shown via the network with solid edges, is an acyclic digraph. Agent $0$ is the leader. The edges of $\mathcal{G}_l$ is shown by the dashed arrow. Here, the leader is the global sink of the $\GG\cup\GG_l$, therefore, its information reaches all the agents in an explicit or implicit manner. }
    \label{fig::DAC}
\end{figure}
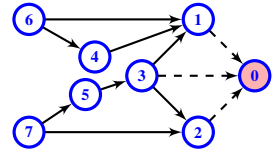

\section{Problem definition}\label{sec::prob_def}
We consider a group of $N$ heterogeneous MAS whose dynamics is described~by
\begin{align}\label{eq::agent_dyn}
    \dvect{x}^i(t)=\vect{A}^i\,\vect{x}^i(t)+\vect{B}^i\,\vect{u}^i(t),\quad i\in\{1,\cdots,N\},
\end{align}
where $\vect{x}^i\in\real^n$ is the state vector and $\vect{u}^i\in\real^{m^i}$ is the control vector. Throughout the paper we assume that the agents' dynamics~\eqref{eq::agent_dyn} is controllable, i.e.,  $(\vect{A}^i,\vect{B}^i)$ for $i\until{N}$ is controllable. These agents (referred hereafter as followers) aim to follow a dynamic signal $\vect{x}^0(t):\real_{\geq0}\to\real^n$ with possibly a locally chosen offset. This signal can be a dynamic reference signal of a virtual leader or the state of an active physical leader with (possibly) a nonlinear dynamics, e.g., $ \dvect{x}^0(t)=f^0(\vect{x}^0(t),\vect{u}^0(t),t)$. The dynamical model and the input $\vect{u}^0\in\real^{m^0}$ of the leader is not known to the followers. The interaction topology between the followers is described by a acyclic digraph, denoted by $\mathcal{G}$.
Only a subset of followers in $\GG$, denoted by $\mathcal{N}_{\text{in}}^0\neq\{\}$, has access to $\vect{x}^0(t)$ at the sampling times $t_k\in\real$, $k\in\mathbb{Z}_{\geq0}$. Throughout the paper we assume that $T_k=t_{k+1}-t_k\in\real_{>0}$ for any $k\in\mathbb{Z}_{\geq0}$ with $t_0=0$. We let $\mathcal{G}_l$ be the digraph consisted of the leader and $\mathcal{N}_{\text{in}}^0$ and the directed edges connecting $\mathcal{N}_{\text{in}}^0$ to the leader. In what follows, we assume that the leader is the global sink of $\overline{\GG}=\GG\cup\GG_l$, so that its information reaches all the agents in an explicit or implicit manner (see Fig.~\ref{fig::DAC} for an example). We let $\overline{\mathcal{N}}_{\text{out}}^{\,i}$ be the set of the out-neighbors of agent $i\in\{0,1,\cdots,N\}$ in graph $\overline{\GG}$; we note the $\overline{\mathcal{N}}_{\text{out}}^0=\{\}$. Finally, we call the followers homogeneous if $(\vect{A}^i,\vect{B}^i)=(\vect{A},\vect{B})$, for $i \until{N}$.

Give that we only have a limited information about the leader (only the sampled states of the leader $\vect{x}^0(t_k)$ is available), we seek a practical solution that enables the followers to arrive at the sampled state of the leader before the next sampling time. Therefore, our objective in this paper is to design a distributed control rule for the input vector $\vect{u}^i(t)$ of each follower $i\in\{1,\cdots,N\}$ such that 
\begin{align}\label{eq::objective}
    \vect{x}^i(t_{k+1})=\vect{x}^0(t_k)-\vect{F}^{i0}(t_k),\quad i\in\{1,\cdots,N\}.
\end{align}
That is, the follower $i\until{N}$ can steer itself to be in $\vect{F}^{i0}(t_k) \in \real^n$ offset with respect to the state $\vect{x}^0(t_k)$ of the leader in time before the next sampling time $t_{k+1}$. We note that the set of offsets $\{\vect{F}^{i0}(t_k)\}_{i=1}^N$, when it is related to the position offsets of the agents, defines the formation of the followers around the leader. Here, the term formation refers to \emph{transnational invariant} formation~\cite[Section 6.1.1]{MM-EM::10}. For scenarios where the objective is to synchronize to the state of the leader, $\vect{F}^{i0}(t_k)$ is set to zero for all $i\in\{1,\cdots,N\}$. To form the offset, we assume that at each sampling time $t_k$, follower $i\in\{1,\cdots,N\}$ knows  $\vect{F}^{ij}(t_k)=\vect{F}^{i0}(t_k)-\vect{F}^{j0}(t_k)$ for $j\in\overline{\mathcal{N}}_{\text{out}}^{\,i}$; either the follower is given  $\vect{F}^{ij}(t_k)$ with respect to its out-neighbor $j$ or constructs it locally after agent $j$ sends its $\vect{F}^{j0}(t_k)$ to agent $i$. We note that if the leader is a global sink of $\overline{\GG}$, given   $\vect{x}^0(t_k)$ and a set of $\vect{F}^{ij}(t_k)$, $i\until{N}$ and $j\in\overline{\mathcal{N}}_{\text{out}}^{\,i}$, we can show that the state offset $\vect{F}^{i0}(t_k)$ for follower $i$ with respect to the leader is unique.

\section{Main result}\label{sec::main}
In this section, we develop a novel distributed solution to solve the leader-following problem stated in Section~\ref{sec::prob_def}. To present this result, we recall that 
\begin{align}\label{eq::G}
    \vect{G}(t)=&\int_{0}^{t} {\ee^{\vect{A}(t- \tau )}\vect{B}\vect{B}^\top\ee^{\vect{A}^\top(t-\tau)}}d\tau,
\end{align}
is the controllability Gramian of $(\vect{A},\vect{B})$ for any finite time $t\in\real_{>0}$. Since $(\vect{A},\vect{B})$ is controllable, $\vect{G}(t)$ is full rank and invertible at each time $t\in\real_{>0}$. We start by using a classical optimal control result to make the following statement. 

\begin{lem}\label{lem::classical}
\rm Consider a leader-following with an offset problem where each follower's
dynamics is given by~\eqref{eq::agent_dyn} with $(\vect{A}^i,\vect{B}^i)$ controllable. Suppose $i$ is a follower in $\GG$ that has access to $\vect{x}^0(t)$ of the leader at each sampling time $t_k$, $k\in\mathbb{Z}_{\geq0}$, i.e., $i\in\mathcal{N}^0_{\text{in}}$. Also, $\vect{F}^{i0}(t_k) \in \real^n$ is the desired state offset with respect to $\vect{x}^0(t_k)$. Starting at an initial condition $\vect{x}^i(t_0)\in\real^n$ with $\vect{u}^i(t_0)=\vect{0}$, for any $i\in\mathcal{N}^0_{\text{in}}$ let 
\begin{align}\label{eq::ctrlA1}
      \vect{u}^i(t)\!\!=\vect{B}^{i^\top}\! \!\!\ee^{\vect{A}^{i^\top}\!\!(t_{k+1}-t)}\vect{G}_k^{i^{-1}}\!\!(\vect{x}^0(t_k)\!\!-\vect{F}^{i0}(t_k)\!\!-\ee^{\vect{A}^i T_k}\vect{x}^{i}(t_k)),
    \end{align} for $t\in(t_k,t_{k+1}]$, where
    where $T_k=t_{k+1}-t_k\in\real_{>0}$, and 
\begin{align}\label{eq::G_k}
    \vect{G}^i_k=\vect{G}^i(T_k)=&\int_{0}^{T_k} {\ee^{\vect{A}^i(T_k- \tau )}\vect{B}^i\vect{B}^{i^\top}\ee^{\vect{A}^{i^\top}(T_k-\tau)}}d\tau.
\end{align}
Then, for every $i\in\mathcal{N}^0_{\text{in}}$ we have $\vect{x}^i(t_{k+1})=\vect{x}^0(t_k)-\vect{F}^{i0}(t_k)$ for all $k\in\mathbb{Z}_{\geq0}$. Moreover, at each time $t\in[t_k,t_{k+1}]$, the control input $\vect{u}^i(t)$ of $i\in\mathcal{N}^0_{\text{in}}$ satisfies
\begin{subequations}\label{eq::optimal_control}
\begin{align}
    \vect{u}^i(t)=&\argmin \int_{t_k}^{t_{k+1}}\vect{u}^i(\tau)^\top\vect{u}^i(\tau)\text{d}\tau,\quad \text{subject to}\\
    &\dvect{x}^i(t)=\vect{A}^i\,\vect{x}^i(t)+\vect{B}^i\,\vect{u}^i(t),\\
    &\vect{x}^i(t_k)=\vect{x}^i(t_k),~~ \vect{x}^i(t_{k+1})=\vect{x}^0(t_k)-\vect{F}^{i0}(t_k).
\end{align}
\end{subequations}
\end{lem}
\begin{proof}
The proof follows from the classical finite time minimum energy optimal control design~\cite[page 138]{FL-DV-VS:12}.
\end{proof}

Lemma~\ref{lem::classical} essentially states that any follower that samples the leader, in the inter-sampling time interval can use the classical minimum energy control to steer towards the latest sampled state of the leader (with an offset if specified). Next, we show that this idea can be extended to a distributed setting in which only a subset of the followers have access to the leader's sampled state. To present our results we first introduce some notations. We denote the adjacency matrix and out-degree matrix of the followers' interaction topology $\mathcal{G}$, respectively, by $\vectsf{A}=[\mathsf{a}_{ij}]$ and  $\vectsf{D}_{\text{out}}=\Diag{\dout^1,\dout^2,\cdots,\dout^N}$. We let 
\begin{align} 
    \mathit{1}^i=\begin{cases}
    1,&i\in\mathcal{N}^0_{\text{in}},\\
    0,&\text{otherwise},
    \end{cases}
\end{align}
be the indicator operator that defines the state of connectivity of follower $i$ to the leader. For $i\until{N}$, we also define
\begin{subequations}
\begin{align}
&\vect{P}^i(t)=\begin{cases}\vect{0}&~~t=t_k,\\
\overline{\vect{G}}^{\,i^{-1}}_k(t)&~~t\in(t_k,t_{k+1}],\end{cases}\quad\quad \text{where}\label{eq::P}\\
  &  \overline{\vect{G}}^{\,i}_k(t)=
    \int_{t_k}^t {\ee^{\vect{A}^i(t - \tau )}\vect{B}^i\vect{B}^{i^\top} \ee^{\vect{A}^{i^\top}(t_{k+1}-\tau)}}d\tau, \quad \!\!
     t \in [t_k,t_{k+1}].\label{eq::barG_k}
\end{align}
\end{subequations}
We notice that $\overline{\vect{G}}^{\,i}_k(t)=\vect{G}^i(t-t_k)\,\ee^{\vect{A}^{i^\top} (t_{k+1}-t)}$
, where $\vect{G}^i$ is the controllability Gramian~\eqref{eq::G}. Therefore at each finite time $t\in(t_k,t_{k+1}]$, by virtue of controllability of $(\vect{A}^i,\vect{B}^i)$, $\overline{\vect{G}}^{\,i}_k(t)$ is invertible. Moreover, note that using the classical control results we can show that $\overline{\vect{G}}^{\,i}_k(t)$ can be computed numerically from 
$\overline{\vect{G}}^{\,i}_k(t)=\vect{W}^i(t)\vect{\Phi}^i(t)$ where $\vect{W}^i(t)=\vect{G}^i(t-t_k)$ and $\vect{\Phi}^i(t)=\ee^{\vect{A}^{i^\top}(t_{k+1}-t)}$ for $t\in[t_k,t_{k+1}]$ are obtained from
\begin{align*}
      \dot{\vect{W}}^i(t)&=\vect{A}^i\vect{W}^i(t)+\vect{W}^i(t)\vect{A}^{i^\top}\!\!+\vect{B}^i\vect{B}^{i^\top}\!\!, \,\,\,\,\vect{W}^i(t_k)=\vect{0}_{n\times n},\\
      \dot{\vect{\Phi}}^i(t)&=-\vect{A}^{i^\top}\vect{\Phi}^i(t),\,\qquad\qquad\qquad\qquad\,\,\vect{\Phi}^i(t_k)=\ee^{\vect{A}^{i^\top} \!\!T_k}.
\end{align*}

With the proper notations at hand, we present our distributed solution to solve our leader-following problem of interest as follows.

\begin{thm}[A leader-following algorithm for a group of heterogeneous LTI followers]\label{thm::main}
{\rm
Consider a leader-following problem where the followers' dynamics are given by ~\eqref{eq::agent_dyn}. Suppose the leader's time-varying state is $\vect{x}^0:\real_{\geq0}\to\real^n$.
Let the network topology $\overline{\mathcal{G}}=\mathcal{G}\cup\mathcal{G}_l$ be an acyclic digraph with leader, node $0$, as the global sink. Suppose 
every follower $i\in\mathcal{N}^0_{\text{in}}$ has access to $\vect{x}^0(t)$ at each sampling time $t_k$, $k\in\mathbb{Z}_{\geq0}$. Let $\vect{F}^{i0}(t_k) \in \real^n$ and $\vect{F}^{ij}(t_k) \in \real^n$ be the desired state offset (formation) with reference to $\vect{x}^0(t_k)$ and $\vect{x}^j(t_{k+1})$, respectively. 
Starting at an initial condition $\vect{x}^i(t_0)\in\real^n$ with $\vect{u}^i(t_0)=\vect{0}$, let for $t\in(t_k,t_{k+1}]$
\begin{align}\label{eq::ctrlA2}
    \vect{u}^i(t)=&\,\omega_l\Big(\vect{B}^{i^\top} \ee^{\vect{A}^{i^\top}(t_{k+1}-t)}\vect{G}_k^{i^{-1}}\times\nonumber\\
    &\qquad\qquad\qquad(\vect{x}^0(t_k)-\vect{F}^{i0}(t_k)-\ee^{\vect{A}^i T_k}\vect{x}^{i}(t_k))\Big)+\nonumber\\
    &\, \omega_f\Big(\vect{B}^{i^\top} \ee^{\vect{A}^{i^\top}(t_{k+1}-t)}{\vect{G}}_k^{i^{-1}}\times\nonumber\\ &\qquad\qquad\sum\limits_{j = 1}^N \mathsf{a}_{ij}\vect{G}^j_k\vect{P}^j(t)(\vect{x}^j(t)-\ee^{\vect{A}^j(t-t_k)}\vect{x}^j(t_k))\nonumber\\
    & \!\!\!\!\!\!+\vect{B}^{i^\top} \ee^{\vect{A}^{i^\top}(t_{k+1}-t)}\vect{G}_k^{i^{-1}}\times\nonumber\\
    &\quad\sum\limits_{j = 1}^N {\mathsf{a}_{ij}(\ee^{\vect{A}^j T_k}\vect{x}^j(t_k)-\ee^{\vect{A}^i T_k}\vect{x}^{i}(t_k)-\vect{F}^{ij}(t_k))}\Big),
\end{align}
where $\vect{P}^j(t)$ is given in~\eqref{eq::P}, $\omega_l^i=\frac{\mathit{1}^i}{\mathit{1}^i+\dout^i}$, and $\omega_f^i=\frac{1}{\mathit{1}^i+\dout^i}$. Then, the followings hold for $t\in\real_{\geq0}$ and $k\in\mathbb{Z}_{\geq0}$:
\begin{itemize} \setlength\itemsep{1em}
    \item[(a)] $\vect{x}^i(t_{k+1})=\vect{x}^0(t_k)-\vect{F}^{i0}(t_k)$, moreover, $\vect{x}^j(t_{k+1})-\vect{x}^i(t_{k+1})=\vect{F}^{ij}(t_k)$ $i,j\until{N}$ and $i\neq j$; 
    \item[(b)] the trajectory of every follower $i\until{N}$ is
    \begin{align}\label{eq::agent_traj}
    \vect{x}^i(t)&=\ee^{\vect{A}^i(t-t_k)}\vect{x}^i(t_k)+\\
    &\quad\overline{\vect{G}}\,^i_k(t)\vect{G}_k^{i^{-1}}(\vect{x}^0(t_k)\!-\!\vect{F}^{i0}(t_k)\!-\!\ee^{\vect{A}^iT_k}\vect{x}^{i}(t_k));\nonumber
        \end{align}
        \item[(c)] the control input $\vect{u}^i(t)$ of every agent $i\until{N}$ is equal to~\eqref{eq::ctrlA1}.\boxend
\end{itemize}
}
\end{thm}
The proof of Theorem~\ref{thm::main} is given in the appendix. Several observations and remarks are in order regarding the leader-following algorithm of  Theorem~\ref{thm::main}.

\begin{rem}[Robustness to state perturbations]\label{rem::robustness}
{\rm We observe that the leader-following algorithm of Theorem~\ref{thm::main} has robustness to state perturbations similar to the well-known Model Predictive Control (MPC). Even though the controller implemented in each epoch $(t_k,t_{k+1}]$ is an open-loop control, since every follower exerts its state at time $t_k$ as initial condition to the controller, the algorithm can account for the slight perturbations in the follower final state $\vect{x}^i(t_{k+1})$ at the end of each epoch.} \boxend
\end{rem}

\begin{rem}[Implementation of control law~\eqref{eq::ctrlA2}]{\rm
To implement ~\eqref{eq::ctrlA2}, we note that the component of ~\eqref{eq::ctrlA2} that multiplies $\omega_l^i$ is computed using the local variables of follower $i$ and the sampled state of the leader if $\omega_l^i$ is non-zero, i.e., $i\in\mathcal{N}^0_{\text{in}}$. The component of ~\eqref{eq::ctrlA2} that multiplies $\omega_f^i$ is computed using the local variables of follower $i$ and variables of its out-neighbors if $\dout^i\neq0$. To compute this term, if the follower $i$ knows  $(\vect{A}^j,\vect{B}^j)$ of its out-neighbor $j$ (which is the case e.g., when the group is homogeneous), it can implement control ~\eqref{eq::ctrlA2} by obtaining the state $\vect{x}^j$ of its out-neighbor $j$ and computing $\vect{P}^j$ and $\ee^{\vect{A}^j(t-t_k)}$ locally. Otherwise, each follower needs to obtain $\vect{z}^j(t)=\vect{G}^j_k\vect{P}^j(t)(\vect{x}^j(t)-\ee^{\vect{A}^j(t-t_k)}\vect{x}^j(t_k))$ and $\ee^{\vect{A}^jT_k}\vect{x}^j(t_k)$ of its each out-neighbor $j$. We note here that since the interval $(t_k,t_{k+1}]$ is open from the left and the dynamics of all the followers is controllable,  $\vect{P}^i(t)$ is well defined. However, for  $t\to t_k^{+}$ from the right, $\vect{P}^i(t)$ goes to infinity. But, since $(\vect{x}^j(t)-\ee^{\vect{A}^j(t-t_k)}\vect{x}^j(t_k))$ goes to zero as $t\to t_k^+$ from the right, the product $\vect{P}^j(t)(\vect{x}^j(t)-\ee^{\vect{A}^j(t-t_k)}\vect{x}^j(t_k))$ goes to  zero as $t\to t_k^{+}$ from the right. The ``high-gain-challenge" observed here is often a common feature in any approach that is geared towards regulation in prescribed finite time. For example, finite-horizon optimal controls with a terminal constraint inevitably yield gains that go to infinity (see e.g., \cite{SY-WY-HJ-KM::17,YT-KZ-PE::18,LZ-HDC-XC::18}).  In practice, the multiplication of very large and very small values can create numerical problems. To address the problem, one way proposed in the literature is equivalent to employ a deadzone on $\vect{P}^j(t)(\vect{x}^j(t)-\ee^{\vect{A}^j(t-t_k)}\vect{x}^j(t_k))$ at the beginning of each time interval. Another approach is equivalent to using a lager interval $(t_k-\delta,t_{k+1}]$, where $\delta \in\real_{>0}$ is small positive number, to compute $\vect{P}^j(t)$ such that $\vect{P}^j(t)$ for $t\in(t_k,t_{k+1}]$ is no longer goes unbounded when $t$ goes to $t_k^{+}$ from the right. These approaches of course result in somewhat sacrifices in the accuracy at each arrival value at $t_{k+1}$ at the end of time interval $(t_k,t_{k+1}]$. However, as discussed in Remark 4.1, the errors will not accumulate. As interestingly discussed in \cite{SY-WY-HJ-KM::17}, the high-gain-challenge in the finite-time control can be contrasted with non-smooth feedback in the sliding mode control where the gain approaches infinity near sliding surface $x = 0$~\cite{KDY-VIU-UO:99} (but the total control input is zero). However, unlike the sliding mode control, where the practical implementation of the high-gain leads to persistent chattering on the sliding surface~\cite{KDY-VIU-UO:99}, in our case the concern arises only at the start of each transition from one sampling time to the other. The aforementioned practical measures to handle the high-gain-challenge in our setting indeed can be compared to the boundary layer approach~\cite{KDY-VIU-UO:99} in the sliding mode control to eliminate chattering. 
\boxend
}
\end{rem}

\begin{rem}[Time-varying MAS dynamics and network topology]{\rm
From the proof of  Theorem~\ref{thm::main}, we can see that the followers dynamics  can be allowed to be time-varying but piece-wise constant over each time interval $(t_k,t_{k+1}]$, i.e.,
$\vect{A}^i(t)=\Bvect{A}^i(t_k)$ and $\vect{B}^i(t)=\Bvect{B}^i(t_k)$, $i\until{N}$ for  $t\in(t_k,t_{k+1}]$. Similarly the network topology can be allowed to be time-varying as long as between $(t_k,t_{k+1}]$ the topology is fixed and satisfies the connectivity condition of Theorem~\ref{thm::main}.
 \boxend}
\end{rem}

\begin{rem}[Minimum energy control in {$[t_k,t_{k+1}]$}]{\rm
From statement (c) of Theorem~\ref{thm::main} it follows that at each time interval $[t_k,t_{k+1}]$, $k\in\mathbb{Z}_{\geq0}$, the control input $\vect{u}^i$ of each follower $i\until{N}$ is the minimum energy controller that transfers the follower from its current state $\vect{x}^i(t_k)$ to their desired state $\vect{x}^i(t_{k+1})=\vect{x}^0(t_{k})-\vect{F}^{i0}_k(t_k)$. \boxend}
\end{rem}

\begin{rem}[Tracking a priori known desired states at exact sampling time and design of arrival times]{\rm
We note that if the leader is virtual and the sampled states are some desired states that are known a priori to $\mathcal{N}_{\text{in}}^0$ with desired arrival time in $\real_{>0}$, the agents can arrive at the desired state of the leader at the desired arrival time. Furthermore, for the homogeneous followers, in cases that the arrival times is not specified one of the followers in $\mathcal{N}_{\text{in}}^0$ (we refer to it as super node that knows the initial state of all the other followers) can design the arrival times to meet other optimality conditions or to avoid violating constraints such as input saturation. In case of input saturation,  the fact that by virtue of statement (c) of Theorem~\ref{thm::main} the form of input vector of the followers are known to be~\eqref{eq::ctrlA1} can be instrumental to the super node in design of arrival times. Our second demonstrative example in Section~\ref{sec::demo} offers the details.}\boxend
\end{rem}

\begin{rem}[Extension of results to output tracking for a special class of MAS]\label{rem::output_tracking}{\rm
The design methodology of the state offset (formation) algorithm of Theorem~\ref{thm::main} can be used in output tracking for a special class of MAS. Let the network topology be as described in Theorem~\ref{thm::main} and the system dynamics of the followers be~\eqref{eq::agent_dyn} where $\vect{x}^i\in\real^{n^i}$ and $\vect{u}^i\in\real^{m^i}$ (the state and input dimensions of the followers are not necessarily the same). Let the objective be that the output $\vect{y}^i=\vect{C}^i\vect{x}^i\in\real^n$, $n\leq n^i$, of each follower should satisfy
\begin{align}\label{eq::objective-y}
    \vect{y}^i(t_{k+1})=\vect{x}^0(t_k)-\vect{F}^{i0}(t_k),\quad i\in\{1,\cdots,N\}.
\end{align}
If $\vect{C}^i\vect{B}^i$ is full row rank, we can use the control $\vect{u}^i=\vect{B}^{i^\top}\vect{C}^{i^\top}(\vect{C}^i\vect{B}^i\vect{B}^{i^\top}\vect{C}^{i^\top})^{-1}\cdot\\ (\vect{v}^i-\vect{C}^i\vect{A}^i\vect{x}^i)$, $i\in\{1,\cdots,N\}$, to write the output dynamics of each follower $i$ as 
$\dvect{y}^i=\vect{v}^i$. Then the method of Theorem~\ref{thm::main} can be used to design $\vect{v}^i\in\real^n$, which can then be used to obtain the appropriate $\vect{u}^i$ that will make the followers meet~\eqref{eq::objective-y}.}\boxend
\end{rem}

Finally, we note that if the followers are homogeneous, the followers can achieve full synchronization in the sense stated in the result below. 

\begin{cor}[Full synchronization for homogeneous followers]\label{cor::ui=uj}
    {\rm Let the state offset be constant i.e., $\vect{F}^{i0}(t_k)=\vect{F}^{i0}\in\real^n$ for all $i\in\{1,\cdots,N\}$ or (equivalently   $\vect{F}^{ij}(t_k)=\vect{F}^{ij}\in\real^n$ for $i,j\in\{1,\cdots,N\}$), and assume that the followers are homogeneous. Then, it follows from statements (b) and (c) of Theorem~\ref{thm::main} that  the followers' trajectories and inputs satisfy $\vect{x}^j(t)=\vect{x}^i(t)+\vect{F}^{ij}$ for $t\in[t_1,\infty)$ and $\vect{u}^i(t)=\vect{u}^j(t)$ for $t\in(t_1,\infty)$, for every $i,j\until{N}$. One can easily verify this point by shifting the state coordinate with $\vect{F}^{ij}$. Moreover, if the agents are initially in the specified offset i.e., $\vect{x}^j(0)=\vect{x}^i(0)+\vect{F}^{ij}$ for all $i,j\until{N}$, then these qualities also hold for $t\in[0,t_1]$.

Assume that there exists $\vect{K}^i \in \real^{m^i\times n}$, $\vect{W}^i\in \real^{m^i\times m^i}$ for $i\in\{1,\cdots,N\}$ and a controllable pair $(\vect{A},\vect{B})$ known to all followers, such that using $\vect{u}^i=\vect{K}^i\vect{x}^i+\vect{W}^i\vect{v}^i$, $i\in\{1,\cdots,N\}$, makes the followers dynamic homogeneous, i.e., $\dvect{x}^i=\vect{A}\vect{x}^i+\vect{B}\vect{v}^i$, $\vect{A}=\vect{A}^i+\vect{B}^i\vect{K}^i$ and $\vect{B}=\vect{B}^i\vect{W}^i$,  $i\in\{1,\cdots,N\}$. Then, it is also possible to achieve full state synchronization by implementing \eqref{eq::ctrlA2} to $\vect{v}^i$ for heterogeneous followers. One sufficient condition for the existence of $\vect{K}^i$ and $\vect{W}^i$, $i\in\{1,\cdots,N\}$, is that  $\vect{B}^i$ of each follower $i\in\{1,\cdots,N\}$ is full row rank. Then,  $\vect{K}^i=\vect{B}^{i^\top}(\vect{B}^i\vect{B}^{i^\top})^{-1}(\vect{A}-\vect{A}^i)$, $\vect{W}^i=\vect{B}^{i^\top}(\vect{B}^i\vect{B}^{i^\top})^{-1}\vect{B}$ and $(\vect{A},\vect{B})$ can be any controllable pair.}\boxend
\end{cor}

\section{Demonstrative examples}\label{sec::demo}
In this section, we demonstrate our results via numerical examples. 

\subsection{A nonlinear-leader following problem for a group of heterogeneous followers}
Consider a group of $7$ mass-spring-damper system (followers)
\begin{align}\label{eq::lmsd_dnm}
    \dvect{x}^i=\underbrace{\begin{bmatrix}0&1\\-\frac{k^i}{m^i}&-\frac{b^i}{m^i}\\\end{bmatrix}}_{\vect{A}^i}\vect{x}^i+\underbrace{\begin{bmatrix}0\\\frac{1}{m^i}\end{bmatrix}}_{\vect{B}^i}u^i, \quad i \in \{1, \dots, 7\}
\end{align}
where $\vect{x}^i=[x^i \quad \dot{x}^i] \in \reals^2$ is the state vector with $x^i \in \reals$ and $\dot{x}^i \in \reals$ representing the displacement and velocity of the mass, $k^i$, $b^i$ and $m^i$ are spring constant, damping constant and mass, respectively, and $u^i \in \reals$ is the input force. The system's parameters $(k^i,b^i,m^i)$ for $i \in\{1, \dots, 7\}$ are $(1, 0.5, 5)$, $(2, 0.5, 15)$, $(2.5, 1.5, 10)$, $(3, 0.8, 8)$, $(3.5, 1.5, 5)$, $(1.2, 1.8, 12)$, and $(0.5, 1, 10)$, respectively. The leader denoted by $0$ is a nonlinear mass-spring-damper system 
\begin{align}
        \dot{\vect{x}}^0=\begin{bmatrix}\dot{x}^0\\ \frac{1}{m^0}(u^0-b^0\dot{x}^0-k^0x^0-0.6 x^{0^3})\end{bmatrix},
\end{align}
where the input $u^0$ is unknown to the followers and the system parameters $(k^0,b^0,m^0)=(1.2,2,5)$. The interaction topology of the systems is shown in Fig.~\ref{fig::DAC}.
Followers $1$, $2$ and $3$ obtain the state of the leader with a sampling rate of $1$ per second, i.e., $T_k=1$ second, $k\in\mathbb{Z}_{\geq0}$. 
The followers start at  $\vect{x}^1(0)=[0 \quad 0]^\top$, $\vect{x}^2(0)=[-0.5 \quad 0]^\top$, $\vect{x}^3(0)=[-1\quad 0]^\top$, $\vect{x}^4(0)=[-1.5 \quad 0]^\top$, $\vect{x}^5(0)=[-2 \quad 0]^\top$, $\vect{x}^6(0)=[-2.5 \quad 0]^\top$, $\vect{x}^7(0)=[-3 \quad 0]^\top$ in a formation with uniform distance $0.5(m)$ to the previous number of the follower. The objective is for the followers to track the sampled state of the leader while preserving the initial formation of the systems at every sampling time~$t_k$. The follower $i$ only knows the local formation, i.e., $\vect{F}^{ij}(0)$ for $j\in\overline{\mathcal{N}}_{\text{out}}^{\,i}$. For example, follower 3 knows  $\vect{F}^{30}(0)=[1\quad0]^\top$, $\vect{F}^{31}(0)=[1\quad0]^\top$,  and $\vect{F}^{32}(0)=[0.5\quad0]^\top$.

\begin{figure}
  \centering
  \includegraphics[width=0.45\textwidth]{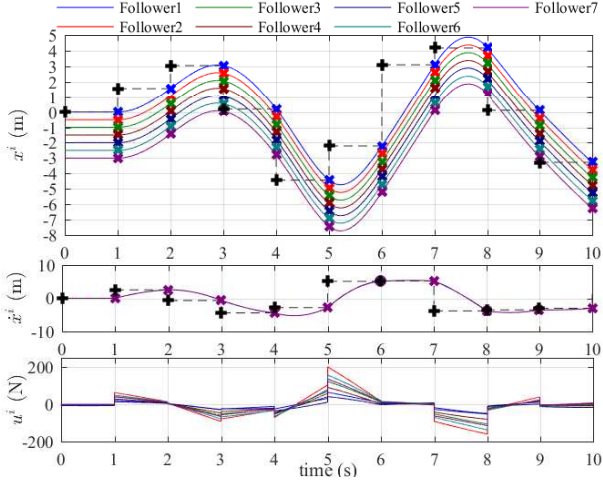}
  \caption{The state and control trajectories of followers of the first numerical example. }
  \label{fig::leader_follower_ex1}
\end{figure}

The result of implementing the algorithm of Theorem~\ref{thm::main} is shown in Fig. \ref{fig::leader_follower_ex1}. The `$+$' represents the sampled leader's states and `$\times$' shows the followers track the sampled state in the desired formation at the next sampled time. In this example interestingly in the transition times similar to what is expected from homogeneous followers the state and input of all the followers are almost offset-synchronized. However, this property is not necessarily true in general for heterogeneous followers.

\subsection{Reference state tracking for a group of second integrator dynamics with bounded inputs} 
We consider a group of $6$ followers with second order integrator dynamics
\begin{align}\label{eq::sec_dnm}
    \dvect{x}^i=\underbrace{\begin{bmatrix}0&1\\0&0\\\end{bmatrix}}_{\vect{A}}\vect{x}^i+\underbrace{\begin{bmatrix}0\\1\end{bmatrix}}_{\vect{B}}u^i,\quad -5\leq u^i\leq5,
\end{align}
for $i\until{6}$. 
The interaction topology of these followers is shown in Fig.~\ref{fig::int-top-1}, where, agent~$0$ is the virtual leader that is defined more precisely below.  Starting at initial conditions $\vect{x}^1(0)=[0 \quad 0]^\top$, $\vect{x}^2(0)=[2 \quad 0]^\top$, $\vect{x}^3(0)=[-2 \quad 0]^\top$, $\vect{x}^4(0)=[5 \quad 0]^\top$, $\vect{x}^5(0)=[10 \quad 0]^\top$, $\vect{x}^6(0)=[-10 \quad 0]^\top$, the leader-following mission for this team is to traverse through the sequence of desired states $\vect{x}^{\text{d}}=\{\vect{x}^{\text{d}}_1,\vect{x}^{\text{d}}_2,\vect{x}^{\text{d}}_3,\vect{x}^{\text{d}}_4\}=\left\{\begin{bmatrix}50 \\ 10\end{bmatrix},\begin{bmatrix}-50 \\ 10\end{bmatrix},\begin{bmatrix}20 \\ 10\end{bmatrix},\begin{bmatrix}0 \\ 0\end{bmatrix} \right\}$, which for privacy reason are only known to follower~$1$. The objective is to meet the sequence of desired states 
without violating any of the followers' control bounds. In this problem setting, follower~$1$ is the super node that knows the initial starting state of all the followers in the team and has computational power to compute the arrival times as follows to meet the team's objective. First, we note that by virtue of statement (c) of Theorem~\ref{thm::main} the form of input vector of the followers are known to be~\eqref{eq::ctrlA1}. Since follower 1 knows $\vect{x}^i(t_0)$ for $i \in \{1,\cdots,6\}$, follower $1$ can evaluate $u^i(t)$ of all the followers. Starting with $t^{\text{d}}_{0}=0$, follower $1$ computes the arrival time at desired state $\vect{x}^{\text{d}}_1$ from the process below
\begin{align}\label{eq::min_time}
    t^{\text{d},i}_{1}=&\argmin \int_{t^{\text{d}}_0}^{t^{\text{d},i}_{1}} \text{d}\tau \quad\text{subject to}~-5\leq u^i(t)\leq 5,
\end{align}
where 
${u}^i(t)=\vect{B}^\top \ee^{\vect{A}^\top (t^{\text{d},i}_{1}-t)}\vect{G}_0^{-1}(\vect{x}^{\text{d}}_0-\ee^{\vect{A}T_0}\vect{x}^{i}(0))$ with $T_0=t^{\text{d},i}_{1}-t^{\text{d}}_{0}$. Then, the arrival time so that the followers input do not saturate over $(t_0^d,t^{\text{d},i}_{1}]$ is set to $t_1^{\text{d}}=\max\{t_{1}^{\text{d},i}\}$. (II) Due to Corollary~\ref{cor::ui=uj}, after first epoch, the followers inputs are equal to each other. Then, the remaining arrival time $t_l^{\text{d}}$, $l\in\{2,3,4\}$ are computed  from the optimization problem 
\begin{align}\label{eq::min_time_2}
    t^{\text{d}}_{k+1}=&\argmin \int_{t^{\text{d}}_k}^{t^{\text{d}}_{k+1}} \text{d}\tau \quad\text{subject to}~-5\leq u(t)\leq 5,
\end{align}
where 
$u(t)=\vect{B}^{\top} \ee^{\vect{A}^{\top}(t^{\text{d}}_{k+1}-t)}\vect{G}_k^{-1}(\vect{x}^{\text{d}}_{k+1}-\ee^{\vect{A}T_k}\vect{x}^{\text{d}}_k)$ with $T_k=t^{\text{d}}_{k+1}-t^{\text{d}}_{k}$, for $k\in\{1,2,3\}$. The solution for this set of sequential optimal control problem is $t_1^{\text{d}}=6.7178$, $t_2^{\text{d}}=25.2061$, $t_3^{\text{d}}=30.1592$ and $t_4^{\text{d}}=40.4885$ seconds.
At the end of process, follower $1$ broadcasts the times to the network. Broadcasting the reference states is not allowed due to privacy reasons. We note that the desired arrival times can be done offline by the system operator. To match the notation in~\eqref{eq::ctrlA2},  at the implantation stage, we set $\vect{x}^0(t_{k-1})=\vect{x}^{\text{d}}_{k}$, $T_{k-1}=t^{\text{d}}_{k}-t^{\text{d}}_{k-1}$, and $t_k=t_{k-1}+T_{k-1}$, $k\until{4}$, where $t^{\text{d}}_{0}=0$. Figures~\ref{fig::ex2} shows that all the followers meet the desired  reference state of the virtual leader at the specified arrival times without delay (the `$+$' marks the reference states). Figure~\ref{fig::ex2} also shows the control history of the agents. As seen, the control inputs respect the saturation bounds $5$ or $-5$. We can also observe that the followers' states and inputs, as predicted in Corollary~\ref{cor::ui=uj}, are all synchronized after the first epoch. We should mention that by virtue of Lemma~\ref{lem::finite_sampling_time} in Appendix B, every $T_k$, $k\in\{0,1,2,3\}$, designed as described above is guaranteed to be a finite value.
 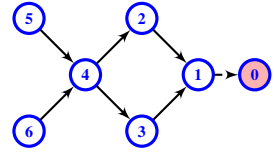
\begin{figure}
    \centering
     \begin{tikzpicture}[auto,thick,scale=0.5, every node/.style={scale=0.67}]
\tikzset{edge/.style = {->,> = latex'}}
el/.style = {inner sep=2pt, align=left, sloped},
\node (0) at (0,0) [draw, minimum size=15pt,color=blue, circle, very thick,fill=red!30] {{\small \textbf{0}}};
\node (1) at (-1.5,0) [draw, minimum size=15pt,color=blue, circle, very thick] {{\small \textbf{1}}};
\node (2) at (-3,1.5) [draw, minimum size=15pt,color=blue, circle, very thick] {{\small \textbf{2}}};
\node (3) at (-3,-1.5) [draw, minimum size=15pt,color=blue, circle, very thick] {{\small \textbf{3}}};
\node (4) at (-4.5,0) [draw, minimum size=15pt,color=blue, circle, very thick] {{\small \textbf{4}}};
\node (5) at (-6,1.5) [draw, minimum size=15pt,color=blue, circle, very thick] {{\small \textbf{5}}};
\node (6) at (-6,-1.5) [draw, minimum size=15pt,color=blue, circle, very thick] {{\small \textbf{6}}};
\draw[edge,dashed]  (1)to  (0);
\draw[edge]  (2)to  (1);
\draw[edge]  (3)to  (1);
\draw[edge]  (4)to  (2);
\draw[edge]  (4)to  (3);
\draw[edge]  (5)to  (4);
\draw[edge]  (6)to  (4);

\end{tikzpicture}
    \caption{An interaction topology with 6 followers. Agent \textbf{0} is the virtual leader.}
    \label{fig::int-top-1}
\end{figure}

\begin{figure}
  \centering
  \includegraphics[width=0.45\textwidth]{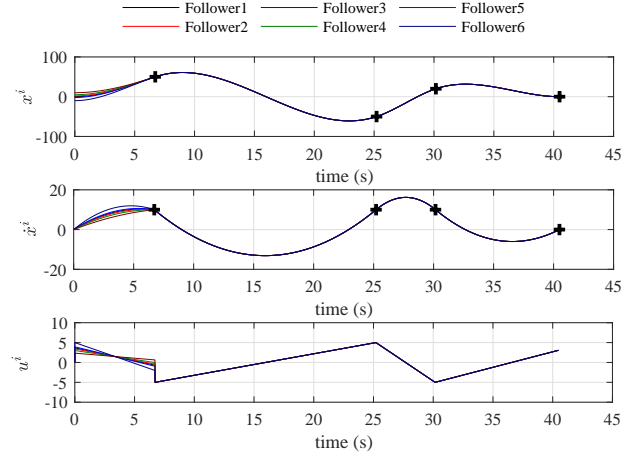}
  \caption{The state and control trajectories of followers of the second numerical example.}
  \label{fig::ex2}
\end{figure}

\begin{figure}
  \centering
  \includegraphics[width=0.45\textwidth]{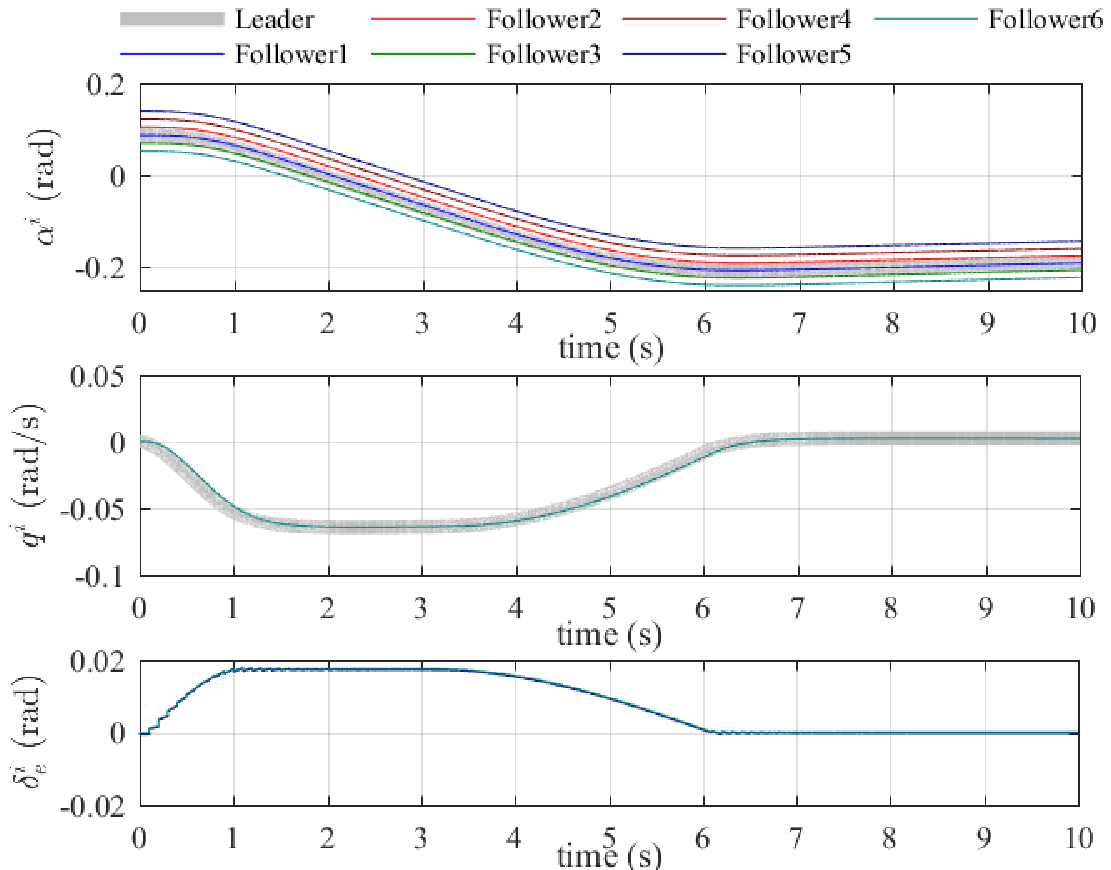}
  \caption{The state and control trajectories of followers of the third numerical example.}
  \label{fig::ex3}
\end{figure}

\subsection{Output tracking for a group of aircraft} 
We consider a group of $7$ aircraft whose short-period dynamics is given by (taken from \cite[Example 10.1]{MVC:07})
\begin{align*}
    \underbrace{\begin{bmatrix}\dot{\alpha}\\\dot{q}\end{bmatrix}}_{\dvect{x}^i}&=\underbrace{\begin{bmatrix}-0.0115&1\\-0.0395&-2.9857\\\end{bmatrix}}_{\vect{A}}\underbrace{\begin{bmatrix}\alpha\\q\end{bmatrix}}_{\vect{x}^i}+\underbrace{\begin{bmatrix}-0.1601\\-11.0437\end{bmatrix}}_{\vect{B}}\underbrace{\delta_e^i}_{u^i},\\ 
    y^i(t)&=\underbrace{\begin{bmatrix}0&1
    \end{bmatrix}}_{\vect{C}} \vect{x}^i,
\end{align*}
where $\alpha^i$, $q^i$ and $\delta_e^i$ are are respectively, angle of attack, pitch rate and elevator angle of aircraft $i\in\{0,\cdots,6\}$. The interaction topology of these aircraft is shown in Fig.~\ref{fig::int-top-1}, where, agent~$0$ is the leader. For this system $\vect{C}\vect{B}=-11.0437$, therefore the condition of Remark~\ref{rem::output_tracking} is satisfied and we can design a distributed algorithm to synchronize the pitch rate of the follower aircraft $\{1,\cdots,6\}$ to the  pitch rate of the leader aircraft when only sampled pitch rate of the leader at every 0.1 seconds is available to the follower aircraft $1$. Figure \ref{fig::ex3} demonstrates the results.

\section{Conclusion}\label{sec::conclusion}
In this paper, we have proposed a distributed leader-following algorithm for heterogeneous multi-agent systems with an active leader with unknown input. We have proved that our distributed leader-following algorithm for the linear followers steers the group to be at the sampled states of the leader at the specified arrival times. We showed that the control input of each follower agent between the sampling times is a minimum energy control. We also showed that after the first sampling epoch, the states of all the homogeneous follower agents are synchronized with each other. We demonstrated our results via  leader-following problems of mass-spring-damper systems, mobile agents with second order integrator dynamics, and a group of aircraft. Future work will focus on extending our results to output following problem.


\bibliographystyle{ieeetr}%
\bibliography{Reference} 

\clearpage
\onecolumn
\section*{Appendix A}
\renewcommand{\theequation}{A.\arabic{equation}}
\renewcommand{\thethm}{A.\arabic{thm}}
\renewcommand{\thelem}{A.\arabic{lem}}
\renewcommand{\thedefn}{A.\arabic{defn}}

\begin{proof}[Proof of Theorem~\ref{thm::main}]
For $\mathcal{G}\cup\mathcal{G}_l$ an acyclic digraph with $0$ as the global sink, the agents can be sorted into a series of hierarchical subsets. Without loss of generality, we sort the agents as follows. Recall that $\VV=\{1,\cdots,N\}$ is the set of the followers. We let $\mathcal{V}_0=\{ 0\}$. Next, we let $\mathcal{V}_1$ to be the subset of agents in $\GG$ that are connected to the leader but they have no out-neighbor in $\GG$, i.e., $\mathcal{V}_1=\{ i\in\VV|\,\,
\mathit{1}^i=1~\text{and}~\mathcal{N}^i_{\text{out}}=\{\}\}$. We sequentially define the lower subset as $\VV_k=\{ i\in \VV\backslash \cup_{j=1}^{k-1} \VV_j|\,\, \mathcal{N}_{\text{out}}^i\subseteq \cup_{j=0}^{k-1} \VV_j\}$, where $k \in \{ 2,\cdots,m\}$, such that $\cup_{j=1}^m\mathcal{V}_j=\mathcal{V}$. In short, in this hierarchy,  the agents in the lower subset only receive information from the agents in the higher subsets.

We use mathematical induction over time intervals $[t_0,t_{k+1}]$, $k \in \mathbb{Z}_{>0}$ for our proof. That is we show that the theorem statements hold for $k=0$. Then assuming that the theorem statements hold for $k$, we show the validity of the statement over $k+1$. The proof of the case for $k=0$ is very similar to the case of $k+1$ and omitted here of brevity. Now let the theorem statements be valid over $[t_0,t_k]$ and we show the validity of the statements at $(t_k,t_{k+1}]$ and as a result the validity of the statement over $[t_0,t_{k+1}]$. For our proof we use as the mathematical induction over $\mathcal{V}_l$ where $l \in \{1, \cdots, m\}$.

Consider first the dynamics of the followers in $\mathcal{V}_l$. For $l=1$, the control \eqref{eq::ctrlA2} reduces to \eqref{eq::ctrlA1}, since $\omega_l^i =1$ and $\sum_{j=1}^N \mathsf{a}_{ij}=0$. Hence statement (c) holds. The trajectory of $\vect{x}^i(t)$ after substituting for the control input $\vect{u}^i$ is
\begin{align*}
        \vect{x}^i(t)&=\ee^{\vect{A}^i(t-t_k)}\vect{x}^i(t_k)+\int_{t_k}^t \ee^{\vect{A}^i(t-\tau)}\vect{B}^i\vect{u}^i(\tau)d\tau\\
        &=\ee^{\vect{A}^i(t-t_k)}\vect{x}^i(t_k)+\int_{t_k}^t \ee^{\vect{A}^i(t-\tau)}\vect{B}^i\vect{B}^{i^\top} \ee^{\vect{A}^{i^\top}(t_{k+1}-\tau)}\vect{G}_k^{i^{-1}}\\
        &\quad\cdot(\vect{x}^0(t_k)-\vect{F}^{i0}(t_k)-\ee^{\vect{A}^i T_k}\vect{x}^{i}(t_k))d\tau\\
        &=\ee^{\vect{A}^i(t-t_k)}\vect{x}^i(t_k)+\overline{\vect{G}}\,^i_k(t)\vect{G}_k^{i^{-1}}(\vect{x}^0(t_k)-\vect{F}^{i0}(t_k)-\ee^{\vect{A}^i T_k}\vect{x}^{i}(t_k)).
\end{align*}
Then given~\eqref{eq::barG_k}, the trajectories of agents $i\in\mathcal{V}_1$ is given by~\eqref{eq::agent_traj} for $t\in\real_{\geq0}$, confirming Statement (b). Moreover, when $t=t_{k+1}$, the final state of the end of this period is
     \begin{align*}
        \vect{x}^i(t_{k+1})&=\ee^{\vect{A}^i T_k}\vect{x}^i(t_k)+\overline{\vect{G}}\,^i_k(t_{k+1})\vect{G}_k^{i^{-1}}(\vect{x}^0(t_k)-\vect{F}^{i0}(t_k)-\ee^{\vect{A}^i T_k}\vect{x}^{i}(t_k))\\
        &=\vect{x}^0(t_k)-\vect{F}^{i0}(t_k).
    \end{align*}
    Also, the relative state with respect to follower $j \in \overline{\mathcal{N}}^{\,i}_{\text{out}}$, is $\vect{x}^j(t_{k+1})-\vect{x}^i(t_{k+1})=\vect{x}^0(t_k)-\vect{F}^{j0}(t_k)-\vect{x}^0(t_k)+\vect{F}^{i0}(t_k)=\vect{F}^{ij}(t_k)$. Therefore,
    statement (a) holds.
    
    Next, let statements (a), (b) and (c) be true for $i \in \VV_{s}$, $s\in \{1,\cdots, l-1\}$. Then, for the follower $i \in \mathcal{V}_l$ we have:
    \begin{align*}
    \vect{x}^i(t)&=\ee^{\vect{A}^i(t-t_k)}\vect{x}^i(t_k)+\int_{t_k}^t \ee^{\vect{A}^i(t-\tau)}\vect{B}^i\vect{u}^i(\tau) d\tau\\
    &=\ee^{\vect{A}^i(t-t_k)}\vect{x}^i(t_k)+\frac{\mathit{1}^i}{\mathit{1}^i+\dout^i}\int_{t_k}^t \ee^{\vect{A}^i(t-\tau)}\vect{B}^i\vect{B}^{i^\top} \ee^{\vect{A}^{i^\top}(t_{k+1}-\tau)}\vect{G}_k^{i^{-1}}\\
    &\quad\cdot(\vect{x}^0(t_k)-\vect{F}^{i0}(t_k)-\ee^{\vect{A}^i T_k}\vect{x}^{i}(t_k))d\tau\\
    &+\frac{1}{\mathit{1}^i+\dout^i}\int_{t_k}^t \ee^{\vect{A}^i(t-\tau)}\vect{B}^i\vect{B}^{i^\top} \ee^{\vect{A}^{i^\top}(t_{k+1}-\tau)}\vect{G}_k^{i^{-1}}\\
    &\quad\cdot\sum\limits_{j = 1}^N \mathsf{a}_{ij}\vect{G}_k^j\vect{P}^j(\tau)(\vect{x}^j(\tau)-\ee^{\vect{A}^j(\tau-t_k)}\vect{x}^j(t_k))d\tau\\
    &+\frac{1}{\mathit{1}^i+\dout^i}\!\!\int_{t_k}^t \!\!\ee^{\vect{A}^i(t-\tau)}\vect{B}^i\vect{B}^{i^\top} \ee^{\vect{A}^{i^\top}(t_{k+1}-\tau)}\vect{G}_k^{i^{-1}}\\
    &\quad\cdot\sum_{j=1}^N{\mathsf{a}_{ij}(\ee^{\vect{A}^j T_k}\vect{x}^j(t_k)\!\!-\!\!\ee^{\vect{A}^i T_k}\vect{x}^{i}(t_k)-\vect{F}^{ij}(t_k))}d\tau\\
    =\,&\ee^{\vect{A}^i(t-t_k)}\vect{x}^i(t_k)+\frac{\mathit{1}^i}{\mathit{1}^i+\dout^i}\overline{\vect{G}}\,^i_k(t)\vect{G}_k^{i^{-1}}(\vect{x}^0(t_k)-\vect{F}^{i0}(t_k)-\ee^{\vect{A}^i T_k}\vect{x}^{i}(t_k))\\
    &+\frac{1}{\mathit{1}^i+\dout^i}\int_{t_k}^t \ee^{\vect{A}^i(t-\tau)}\vect{B}^i\vect{B}^{i^\top} \ee^{\vect{A}^{i^\top}(t_{k+1}-\tau)}\vect{G}_k^{i^{-1}}\\
    &\quad\cdot\sum\limits_{j = 1}^N \mathsf{a}_{ij}\vect{G}_k^j\vect{P}^j(\tau)(\vect{x}^j(\tau)-\ee^{\vect{A}^j(\tau-t_k)}\vect{x}^j(t_k))d\tau\\
    &+\frac{1}{\mathit{1}^i+\dout^i}\overline{\vect{G}}\,^i_k(t)\vect{G}_k^{i^{-1}}\sum\limits_{j = 1}^N \mathsf{a}_{ij}(\ee^{\vect{A}^j T_k}\vect{x}^j(t_k)-\ee^{\vect{A}^i T_k}\vect{x}^{i}(t_k)-\vect{F}^{ij}(t_k)).
    \end{align*}
    Since $j \in \VV_s$, where $s<l$, the trajectory $\vect{x}^j(\tau)$ of agent $j$ is assume to follow \eqref{eq::agent_traj}. Therefore, we can put \eqref{eq::agent_traj} into $\vect{x}^j(\tau)$.
    \begin{align*}
   \vect{x}^i(t)=\,&\ee^{\vect{A}^i(t-t_k)}\vect{x}^i(t_k)+\frac{\mathit{1}^i}{\mathit{1}^i+\dout^i}\overline{\vect{G}}\,^i_k(t)\vect{G}_k^{i^{-1}}(\vect{x}^0(t_k)-\ee^{\vect{A}^i T_k}\vect{x}^{i}(t_k)-\vect{F}^{i0}(t_k))\\
     &\!\!\!\!\!\!\!\!\!\!\!\!\!\!+\frac{1}{\mathit{1}^i+\dout^i}\!\!\int_{t_k}^t \!\!\ee^{\vect{A}^i(t-\tau)}\vect{B}^i\vect{B}^{i^\top} \!\!\ee^{\vect{A}^{i^\top}\!\!(t_{k+1}-\tau)}\vect{G}_k^{i^{-1}}\!\!\sum\limits_{j = 1}^N \mathsf{a}_{ij}\vect{G}_k^j\vect{P}^j(\tau)(\ee^{\vect{A}^j(\tau-t_k)}\vect{x}^j(t_k)\\
    &\quad+\overline{\vect{G}}^j_k(\tau)\vect{G}_k^{j^{-1}}(\vect{x}^0(t_k)-\vect{F}^{j0}(t_k)-\ee^{\vect{A}^j T_k}\vect{x}^{j}(t_k))-\ee^{\vect{A}^j(\tau-t_k)}\vect{x}^j(t_k))d\tau\\
    &+\frac{1}{\mathit{1}^i+\dout^i}\overline{\vect{G}}\,^i_k(t)\vect{G}_k^{i^{-1}}\sum\limits_{j = 1}^N \mathsf{a}_{ij}(\ee^{\vect{A}^j T_k}\vect{x}^j(t_k)-\ee^{\vect{A}^i T_k}\vect{x}^{i}(t_k)-\vect{F}^{ij}(t_k))\\
    =&\ee^{\vect{A}^i(t-t_k)}\vect{x}^i(t_k)+\frac{\mathit{1}^i}{\mathit{1}^i+\dout^i}\overline{\vect{G}}\,^i_k(t)\vect{G}_k^{i^{-1}}(\vect{x}^0(t_k)-\vect{F}^{i0}(t_k)-\ee^{\vect{A}^i T_k}\vect{x}^{i}(t_k))\\
    &+\frac{1}{\mathit{1}^i+\dout^i}\overline{\vect{G}}\,^i_k(t)\vect{G}_k^{i^{-1}}\sum\limits_{j = 1}^N \mathsf{a}_{ij}(\vect{x}^0(t_k)-\vect{F}^{j0}(t_k)\!-\!\ee^{\vect{A}^j T_k}\vect{x}^{j}(t_k))\\
    &+\frac{1}{\mathit{1}^i+\dout^i}\overline{\vect{G}}\,^i_k(t)\vect{G}_k^{i^{-1}}\sum\limits_{j = 1}^N \mathsf{a}_{ij}(\ee^{\vect{A}^j T_k}\vect{x}^j(t_k)-\ee^{\vect{A}^i T_k}\vect{x}^{i}(t_k)-\vect{F}^{ij}(t_k))\\
    =\,&\ee^{\vect{A}^i(t-t_k)}\vect{x}^i(t_k)+\frac{\mathit{1}^i}{\mathit{1}^i+\dout^i}\overline{\vect{G}}\,^i_k(t)\vect{G}_k^{i^{-1}}(\vect{x}^0(t_k)-\vect{F}^{i0}(t_k)-\ee^{\vect{A}^i T_k}\vect{x}^{i}(t_k))\\
    &+\frac{\dout^i}{\mathit{1}^i+\dout^i}\overline{\vect{G}}\,^i_k(t)\vect{G}_k^{i^{-1}}(\vect{x}^0(t_k)-\vect{F}^{i0}(t_k)-\ee^{\vect{A}^i T_k}\vect{x}^{i}(t_k))\\
     =\,&\ee^{\vect{A}^i (t-t_k)} \vect{x}^i(t_k)+\overline{\vect{G}}\,^i_k(t)\vect{G}_k^{i^{-1}}(\vect{x}^0(t_k)-\vect{F}^{i0}(t_k)-\ee^{\vect{A}^i T_k}\vect{x}^{i}(t_k)).
\end{align*}
\begin{align*}
     \vect{x}^i(t_{k+1})=\,&\ee^{\vect{A}^i T_k}\vect{x}^i(t_k)+\overline{\vect{G}}\,^i_k(t_{k+1})\vect{G}_k^{i^{-1}}(\vect{x}^0(t_k)-\vect{F}^{i0}(t_k)-\ee^{\vect{A}^i T_k}\vect{x}^{i}(t_k))\\
     =\,&\vect{x}^0(t_k)-\vect{F}^{i0}(t_k).
\end{align*}
Similarly, the relative state with respect to agent $j \in \overline{\mathcal{N}}\,^i_{\text{out}}$, is $\vect{x}^j(t_{k+1})-\vect{x}^i(t_{k+1})=\vect{x}^0(t_k)-\vect{F}^{j0}(t_k)-\vect{x}^0(t_k)+\vect{F}^{i0}(t_k)=\vect{F}^{ij}(t_k)$. 
Thereby, statement (a) and (b) also hold for the case $l=l$. Then we show that control \eqref{eq::ctrlA2} is equivalent to \eqref{eq::ctrlA1} as follows
\begin{align*}
    \vect{u}^i(t)&=\frac{\mathit{1}^i}{\mathit{1}^i+\dout^i}[\vect{B}^{i^\top} \ee^{\vect{A}^{i^\top} (t_{k+1}-t)}\vect{G}_k^{i^{-1}}(\vect{x}^0(t_k)-\vect{F}^{i0}(t_k)-\ee^{\vect{A}^i T_k}\vect{x}^{i}(t_k))]\\
    &\!\!\!\!\!\!\!\!+\frac{1}{\mathit{1}^i+\dout^i}[\vect{B}^{i^\top} \ee^{\vect{A}^{i^\top}(t_{k+1}-t)}\vect{G}_k^{i^{-1}}\sum\limits_{j = 1}^N \mathsf{a}_{ij}\vect{G}_k^j\vect{P}^j(t) (\vect{x}^j(t)-\ee^{\vect{A}^j (t-t_k)}\vect{x}^j(t_k))\\
    &\quad+\vect{B}^i{^\top} \ee^{\vect{A}^{i^\top}(t_{k+1}-t)}\vect{G}_k^{i^{-1}}\sum\limits_{j = 1}^N {\mathsf{a}_{ij} (\ee^{\vect{A}^j T_k }\vect{x}^j(t_k)-\ee^{\vect{A}^i T_k}\vect{x}^{i}(t_k)-\vect{F}^{ij}(t_k))}]\\
    &=\frac{\mathit{1}^i}{\mathit{1}^i+\dout^i}[\vect{B}^{i^\top} \ee^{\vect{A}^{i^\top} (t_{k+1}-t)}\vect{G}_k^{i^{-1}}(\vect{x}^0(t_k)-\vect{F}^{i0}(t_k)-\ee^{\vect{A}^i T_k}\vect{x}^{i}(t_k))]\\
    &+\frac{1}{\mathit{1}^i+\dout^i}[\vect{B}^{i^\top} \ee^{\vect{A}^{i^\top}(t_{k+1}-t)}\vect{G}_k^{i^{-1}}\sum\limits_{j = 1}^N \mathsf{a}_{ij}\vect{G}_k^j\vect{P}^j(t) (\ee^{\vect{A}^j(t-t_k)}\vect{x}^j(t_k)\\
    &\quad +\overline{\vect{G}}^j_k(t)\vect{G}_k^{j^{-1}}(\vect{x}^0(t_k)-\vect{F}^{j0}(t_k)-\ee^{\vect{A}^j T_k}\vect{x}^{j}(t_k))-\ee^{\vect{A}^j (t-t_k)}\vect{x}^j(t_k))\\
    &+\vect{B}^{i^\top} \ee^{\vect{A}^{i^\top}(t_{k+1}-t)}\vect{G}_k^{i^{-1}}\sum\limits_{j = 1}^N {\mathsf{a}_{ij} (\ee^{\vect{A}^j T_k }\vect{x}^j(t_k)-\ee^{\vect{A}^i T_k}\vect{x}^{i}(t_k)-\vect{F}^{ij}(t_k))}]\\
    &=\frac{\mathit{1}^i}{\mathit{1}^i+\dout^i}[\vect{B}^{i^\top} \ee^{\vect{A}^{i^\top} (t_{k+1}-t)}\vect{G}_k^{i^{-1}}(\vect{x}^0(t_k)-\vect{F}^{i0}(t_k)-\ee^{\vect{A}^i T_k}\vect{x}^{i}(t_k))]\\
    &+\frac{\dout^i}{\mathit{1}^i+\dout^i}[\vect{B}^{i^\top} \ee^{\vect{A}^{i^\top}(t_{k+1}-t)}\vect{G}_k^{i^{-1}}((\vect{x}^0(t_k)-\vect{F}^{i0}(t_k)-\ee^{\vect{A}^i T_k}\vect{x}^{i}(t_k))]\\
    &=\vect{B}^{i^\top}\ee^{\vect{A}^{i^\top} (t_{k+1}-t)}\vect{G}_k^{i^{-1}}(\vect{x}^0(t_k)-\vect{F}^{i0}(t_k)-\ee^{\vect{A}^i T_k}\vect{x}^{i}(t_k)).
\end{align*}
Therefore, statement (c) holds.
Since both the base case $l=1$ and the inductive step have been proved, by mathematical induction statement (a), (b) and (c) hold for all $l \in \{1, \cdots, m\}$.
\end{proof}

\section*{Appendix B}
\renewcommand{\theequation}{B.\arabic{equation}}
\renewcommand{\thethm}{B.\arabic{thm}}
\renewcommand{\thelem}{B.\arabic{lem}}
\renewcommand{\thedefn}{B.\arabic{defn}}
The result below ensures the feasibility of the sampling time design of the second numerical example of Section~\ref{sec::demo}.
\begin{lem}\label{lem::finite_sampling_time}{\rm
Consider a second order integrator system initialized at  $\vect{x}(t_0)=\vect{\chi}(t_0)\in\real^2$ at time $t_0\in\real_{\geq0}$. This system implements the minimum energy controller 
\begin{align}\label{eq::min_en_ctrl}
      {u}(t)=\vect{B}^{\top}\! \ee^{\vect{A}^{\top}(t_{k+1}-t)}\vect{G}(T_k)^{{-1}}(\vect{\chi}(t_{k+1})-&\ee^{\vect{A} T_k}\vect{x}(t_k)),\quad t\in(t_k,t_{k+1}],
    \end{align}  
    and $u(t_0)=0$ to traverse sequentially through a set of $m+1$ points  $\{\vect{\chi}(t_k)\}_{k=0}^{m}\subset \real^2$, where $t_k\in\real_{\geq0}$ is the arrival time at point $\vect{\chi}(t_k)$,   $T_k=t_{k+1}-t_k\in\real_{>0}$, $\vect{A}$ and $\vect{B}$ are given in \eqref{eq::sec_dnm}, and $\vect{G}$ is defined in \eqref{eq::G}.  For this system, there always exists a set of finite arrival times $\{t_k\}_{0}^{m}$ such that  
$|u(t)|\leq |u_{max}|$ for any $t\in[t_0, t_{m}]$, where $u_{max}\in\real_{\geq0}$ is the known bound on the control input.}
\end{lem}

\begin{proof}To establish the proof,
similar to the proof of Theorem~\ref{thm::main}, we relay on the mathematical induction over time intervals $[t_0,t_{k+1}]$, $k \in \{0,\cdots,m-1\}$. The proof of the case for $k=0$ is similar to the case of $k+1$ and omitted here for brevity. Now let the statements be valid over $[t_0,t_k]$. Next, we show the validity of the statements at $(t_k,t_{k+1}]$, and as a result the validity of the statement over $[t_0,t_{k+1}]$. Let $\vect{\chi}(t_k)=[\chi_{k,1}^{} \quad \chi_{k,2}^{}]^\top\in\real^2$, $k\in\{0,\cdots,m\}$. Also, given $t\in(t_k,t_{k+1}]$, let $t'=t-t_k \in (0,T_k]$. Disregarding the control bounds,~\eqref{eq::min_en_ctrl} results in $\vect{x}(t_k)=\vect{\chi}(t_{k})$ for $k\in\{0,\cdots,m-1\}$. Therefore, control~\eqref{eq::min_en_ctrl} can also be expressed as
\begin{align*}
    u(t')=\begin{bmatrix}\frac{12}{T_k^3}(T_k-t')-\frac{6}{T_k^2} & -\frac{6}{T_k^2}(T_k-t')+\frac{4}{T_k} \end{bmatrix}\begin{bmatrix}\chi_{k+1,1}^{}-\chi_{k,1}^{}+T_k \chi_{k,2}^{} \\ \chi_{k+1,2}^{}-\chi_{k,2}^{}\end{bmatrix}.
\end{align*}
Since $u(t')$ is an affine function of $t'$, the maximum value of $|u(t')|$ is at either $t'\to 0^
+$ or $t'=T_k$. That is,  $|u(t')|\leq |u(t'\!\to\! 0^+)|$ or $|u(t')|\leq |u(T_k)|$ where $u(t'\!\to\! 0^+)=\lim_{t'\to 0^{+}}u(t')$.  Next, we  show that there always exists a $T_k$ that makes $|u(t'\to 0^+)|\leq|u_{max}|$ and $|u(T_k)|\leq\!|u_{max}|$, which means that $|u(t')|\leq\! |u_{max}|$, $t'\in(0,T_k]$. Note that 
$\left|u(t'\to0^+)\right|=\left|\frac{6}{T_k^2}(\chi_{k+1,1}^{}-\chi_{k,1}^{})-\frac{2}{T_k}(\chi_{k+1,2}^{}-2\chi_{k,2}^{})\right|\leq \left|\frac{6}{T_k^2}(\chi_{k+1,1}^{}-\chi_{k,1}^{})\right|+\left|\frac{2}{T_k}(\chi_{k+1,2}^{}-2\chi_{k,2}^{})\right|$, 
and $\left|u(T_k)\right|=\Big{|}-\frac{6}{T_k^2}(\chi_{k+1,1}^{}-\chi_{k,1}^{})+$ $\frac{2}{T_k}(2\chi_{k+1,2}^{}-\chi_{k,2}^{})\Big{|}\leq \left|\frac{6}{T_k^2}(\chi_{k+1,1}^{}-\chi_{k,1}^{})\right|+\left|\frac{2}{T_k}(2\chi_{k+1,2}^{}-\chi_{k,2}^{})\right|$.
Since the upper bounds established for $\left|u(t'\to0^+)\right|$ and $\left|u(T_k)\right|$ monotonically decrease when $T_k$ increases, there always exists finite value of $T_k$ such that these upper bounds become equal to $u_{max}$. Therefore, there always exists a finite value of $T_k$ for which control law~\eqref{eq::min_en_ctrl} does not violate the controller saturation bound.
 \end{proof}

\end{document}